\documentclass[a4paper,runningheads]{llncs}

\usepackage[main=english,czech]{babel}
\usepackage[utf8]{inputenc} 

\title{On the m-eternal Domination Number of Cactus Graphs
\thanks{V. Bla\v zej and T. Valla acknowledge the support of the OP VVV MEYS funded project
CZ.02.1.01/0.0/0.0/16\_019/0000765 ``Research Center for Informatics''.}
}
\author{V{\'a}clav Bla{\v z}ej \and Jan Maty{\'a}{\v s} K{\v r}i{\v s}{\v t}an \and Tom{\'a}{\v s} Valla}
\institute{%
Faculty of Information Technology, Czech Technical University in Prague,\\ Prague, Czech Republic
}

\usepackage{todonotes}

\usepackage{amssymb} 
\usepackage{mathtools}
\usepackage{algorithm}
\usepackage{algpseudocode}
\usepackage{thmtools} 
\usepackage{thm-restate}
\usepackage{etoolbox}
\usepackage{caption}

\usepackage{tikz} 
\usetikzlibrary{shapes,arrows}
\tikzstyle{block} = [rectangle, draw, text width=8em, text centered, rounded corners, minimum height=3em]
\tikzstyle{pic} = [rectangle, draw, text centered, rounded corners, minimum height=3em, text width=13em]
\tikzstyle{line} = [draw, -latex']

\def\spread[#1,#2]{#1_1, #1_2, \dots, #1_{#2} } 

\algnewcommand{\LineComment}[1]{\State \(\triangleright\) #1}

\DeclarePairedDelimiter{\ceil}{\lceil}{\rceil}
\DeclarePairedDelimiter{\floor}{\lfloor}{\rfloor}

\newcommand{\upperbound}{\noindent{\bf Upper bound: }}
\newcommand{\lowerbound}{\noindent{\bf Lower bound: }}

\newif\iffullpaper
\fullpapertrue

\begin{document}

\maketitle


\def\spread[#1,#2]{#1_1, #1_2, \dots, #1_{#2} } 

\hyphenation{de-ve-loped}
\hyphenation{com-bi-na-to-ri-al}
\hyphenation{Fried-gut}

\newcommand{\BC}{BC} 
\newcommand{\CONF}{C}
\newcommand{\CONFS}{\mathbb{C}}
\newcommand{\FF}{\mathbb{F}}
\newcommand{\SG}{\mathcal{G}}
\newcommand{\SI}{\mathcal{I}}
\newcommand{\SJ}{\mathcal{J}}
\newcommand{\EDE}{\gamma^{\infty}_\mathrm{me}} 
\newcommand{\EDN}{\gamma^{\infty}_\mathrm{m}} 
\newcommand{\EGC}{\Gamma^{\infty}_\mathrm{m}} 
\newcommand{\DN}{\gamma} 

\newtheorem{observation}{Observation}
\newtheorem{reduction}{Reduction}

\newcommand\MED{m-eternal domination}

\newcommand{\toappendix}[1]{%
  \gappto{\appendixProofText}{
    {#1}
   }
}

\begin{abstract}
    Given a graph $G$, guards are placed on vertices of $G$.
    Then vertices are subject to an infinite sequence of attacks
    so that each attack must be defended by a guard moving from a neighboring vertex.
    The m-eternal domination number is the minimum number of guards such that
    the graph can be defended indefinitely.
    In this paper we study the m-eternal domination number of cactus graphs,
    that is, connected graphs where each edge lies in at most two cycles,
    and we consider three variants of the m-eternal domination number:
    first variant allows multiple guards to occupy a single vertex,
    second variant does not allow it, and in the third variant additional ``eviction'' attacks must be defended.
    We provide a new upper bound for the m-eternal domination number of cactus graphs,
    and for a subclass of cactus graphs called
    Christmas cactus graphs, where each vertex lies in at most two cycles,
    we prove that these three numbers are equal.
    Moreover, we present a linear-time algorithm for computing them.
\end{abstract}

\section{Introduction}

Let us have a graph $G$ whose vertices are occupied by guards.
The graph is subject to an infinite sequence of vertex attacks.
The guards may move to any neighboring vertex after each attack.
After moving, a vertex attack is defended if the vertex is occupied by a guard.
The task is to come up with a strategy such that the graph can be defended indefinitely.

Defending a graph from attacks using guards for an infinite number of steps was introduced by Burger et~al.~\cite{ineq1}.
In this paper we study the concept of the \MED, which was introduced by Goddard et~al.~\cite{eternal-security-in-graphs} (eternal domination was originally called eternal security).


The m-eternal guarding number $\EGC(G)$ is the minimum number of guards which tackle all attacks in $G$ indefinitely.
Here the (slightly confusing) notion of the letter ``m'' emphasizes that multiple guards may move during each round.
The m-eternal domination number $\EDN(G)$ is the minimum number of guards which tackle all attacks indefinitely,
with the restriction that no two guards may occupy a single vertex simultaneously.
We also introduce the m-eternal domination number with eviction $\EDE(G)$, which is similar to $\EDN$ with
the additional requirement, that during each round one can decide to either attack a vertex or choose an ``evicted'' vertex or edge, which must be cleared of guards in the next round.
There is also a variant of the problem studied by Goddard et~al.~\cite{eternal-security-in-graphs} where only one guard may move during each round, which
is not considered in our paper.
We will define all concepts formally at the end of this section.

Goddard et~al.~\cite{eternal-security-in-graphs} established $\EDN$ exactly for paths, cycles, complete graphs and complete bipartite graphs,
showing that $\EDN(P_n) = \ceil{n/2}$, $\EDN(C_n) = \ceil{n/3}$, $\EDN(K_n) = 1$ and $\EDN(K_{m,n}) = 2$.
The authors also provide several bounds for general graphs,
most notably $\DN(G) \leq \EDN(G) \leq \alpha(G)$,
where $\alpha(G)$ denotes the size of the maximum independent set in $G$
and $\DN(G)$ is the size of the smallest dominating set in $G$.
Since that several results focused on finding bounds of $\EDN$ in different conditions or graph classes.

Henning, Klostermeyer and MacGillivray~\cite{bounds-meternal-dom} explored the relationship
between $\EDN$ and the minimum degree $\delta(G)$ of a graph $G$:
If $G$ is a connected graph with minimum degree $\delta(G) \geq 2$ and has $n \neq 4$ vertices,
then $\EDN(G) \leq \floor{(n - 1)/2}$, and this bound is tight.

Finbow, Messinger and van Bommel~\cite{3n-grids} proved the following result for $3 \times n$ grids.
For $n \geq 2$,\\
$$\EDN(P_3 \square P_n) \leq \ceil{6n/7} + \begin{cases}
  1 & \text{if } n \equiv 7,8,14 \text{ or } 15 \bmod 21,\\
  0 & \text{otherwise.}
\end{cases}
$$
Here $G \square H$ denotes the Cartesian product of graphs $G$ and $H$.

Van Bommel and van Bommel~\cite{5n-grids} showed for $5 \times n$ grids that
$$ \floor[\bigg]{ \frac{6n + 9}{5} } \leq \EDN(P_5 \square P_n) \leq \floor[\bigg]{ \frac{4n + 4}{3} }.$$

For a good survey on other related results and topics see Klostermeyer and Mynhardt~\cite{survey-article}.

Very little is known regarding the algorithmic aspects of m-eternal domination.
The decision problem (asking if $\EDN(G)\le k$) is obviously NP-hard and belongs to EXPTIME, however, it is not known
whether it lies in the class PSPACE (see \cite{survey-article}).
On the positive side, there is a linear algorithm for computing $\EDN$ for
trees by Klostermeyer and MacGillivray~\cite{eternal-dom-sets}.
Braga, de Souza and Lee~\cite{proper-interval-graphs} showed that $\EDN(G) = \alpha(G)$ in all proper-interval graphs.
Very recently Gupta et~al.~\cite{intervals} showed that
the maximum independent set in an interval graph on $n$ vertices
can be solved in time $\mathcal{O}(n \log n)$, or $\mathcal{O}(n)$ in the case
when endpoints of the intervals are already sorted.
We can thus compute $\EDN(G)$ efficiently on proper-interval graphs.

In this paper we contribute to the positive side
and provide an extension of the result by Klostermeyer and MacGillivray~\cite{eternal-dom-sets}.
\emph{Cactus} is a graph that is connected and its every edge lies on at most one cycle.
An equivalent definition is that it is connected and any two cycles have at most one vertex in common.
\emph{Christmas cactus graph} is a cactus in which each vertex is in at most two $2$-connected components.
Christmas cactus graphs were introduced by Leighton and Moitra~\cite{christmas-cactus}
in the context of greedy embeddings, where Christmas cactus graphs play an important
role in the proof that every polyhedral graph has a greedy embedding in the Euclidean plane.

Our main result is summarized in the following theorem.
\begin{restatable}{theorem}{cactusalgoritm}
  \label{thm:cactus-algorithm}
  Let $G$ be a Christmas cactus graph.
  Then $\EGC(G) = \EDN(G) = \EDE(G)$ and there exists a linear-time algorithm which computes these values.
\end{restatable}
Using Theorem~\ref{thm:cactus-algorithm} we are able to devise a new bound
on the m-eternal domination number of cactus graphs, which is stated in
Theorem~\ref{thm:cactus-upperbound} in Section~\ref{sec:bound}.
In Section~\ref{sec:algorithm} we provide the linear-time algorithm for computing $\EDN$ of Christmas cactus graphs.

\bigskip
Let us now introduce all concepts formally.
For an undirected graph $G$ let a \emph{configuration} be a multiset $\CONF=\{\spread[c,k]:c_i \in V(G)\}$.
We will refer to the elements of configurations as \emph{guards}.
\emph{Movement} of a guard $u\in\CONF$ means changing $u$ to some element $v\in N_G[u]$ of its closed neighborhood and we denote it by $u \rightarrow v$.
Two configurations $\CONF_1$ and $\CONF_2$ of $G$ are mutually \emph{traversable} in $G$
if it is possible to move each guard of $\CONF_1$ to obtain $\CONF_2$.
A \emph{strategy} in $G$ is a graph $S_G=(\CONFS,\FF)$ where $\CONFS$ is a set of configurations in $G$ of same size and
$\FF = \big\{\{\CONF_1,\CONF_2\}\in\CONFS^2 \mid \hbox{$\CONF_1$ and $\CONF_2$ are mutually traversable in $G$}\big\}$.
The \emph{order} of a strategy is the number of guards in each of its configurations.

We now define the variants of the problem which we study in our paper.
For the purpose of the proof of our main result we devise a variant of the problem, where
a vertex or an edge can be ``evicted'' during a round, that means, no guard may be
left on the respective vertex or edge.
We call the strategy $S_G$ to be \emph{defending against vertex attacks}
if for any $\CONF \in \CONFS$ the configuration $\CONF$ and its neighbors in $S_G$ cover all vertices of $G$,
i.e., when a vertex $v \in V(G)$ is \uv{attacked} one can always respond by changing to a configuration which has a guard at the vertex $v$.
Note that every configuration in a strategy which defends against vertex attacks induces a dominating set in $G$.
We call a strategy $S_G$ to be \emph{evicting vertices} if for any $\CONF \in \CONFS$ and any $u\in V(G)$ the configuration $\CONF$ has a neighbor $\CONF'$ in $S_G$ such that $u\notin \CONF'$,
i.e., when a vertex $v$ is \uv{to be evicted} one can respond by changing to a configuration where no guard is present at $v$.
We call a strategy $S_G$ to be \emph{evicting cycle edges} if for any $\CONF \in \CONFS$ and any edge $\{u,v\}\in E(G)$ lying in some cycle in $G$
the configuration $\CONF$ has a neighbor $\CONF'$ in $S_G$ such that $u,v\notin \CONF'$.
That means, when an edge is \uv{to be evicted} one can respond by changing to a configuration where no guards are incident to the edge.

Let the \emph{m-eternal guard strategy} in $G$ be a strategy defending against vertex attacks in $G$.
Let the m-eternal guard configuration number $\EGC(G)$ be the minimum order among all m-eternal guard strategies in $G$.
Let the \emph{m-eternal dominating strategy} in $G$ be a strategy in $G$
such that none of its configurations has duplicates and is also defending against vertex attacks.
The \emph{m-eternal dominating set} in $G$ is a configuration, which is contained in some m-eternal dominating strategy in $G$.
Let the m-eternal dominating number $\EDN(G)$ be the minimum order of m-eternal dominating strategy in $G$.
Let the \emph{m-eternal dominating strategy with eviction} in $G$ be a strategy such that none of its configurations has duplicates, is defending vertex attacks, is evicting vertices, and is evicting cycle edges in $G$.
Let the m-eternal dominating with eviction number $\EDE(G)$ be the minimum order of m-eternal dominating strategy evicting vertices and edges in $G$.

A cycle in $G$ is a \emph{leaf cycle} if exactly one of its vertices has degree greater than 2.
By $P_n$ we denote a path with $n$ edges and $n + 1$ vertices.
By $G[U]$ we denote the subgraph of $G$ induced by the set of vertices $U \subseteq V(G)$.




%
%
%

\section{The m-eternal domination of Christmas cactus graphs}\label{sec:cacti-graphs}

\iffullpaper
In this section we prove that $\EGC(G) = \EDN(G) = \EDE(G)$ for Christmas cactus graphs by showing the optimal strategy.
The main idea is to repeatedly use reductions of the Christmas cactus graph $G$ to produce smaller Christmas cactus graph $I$.
We prove that the optimal strategy for $G$ uses a constant number of guards more than the optimal strategy for $I$.

This will be one part of the proof of Theorem~\ref{thm:cactus-algorithm}.
Before we describe the reductions in detail, we present several technical tools that are used in the proofs of validity of the reductions.
\fi

\begin{observation}\label{obs:edn-lb}
  Every strategy used in the \MED\ with eviction can be applied in an \MED\ strategy, and every \MED\ strategy can be applied as an m-eternal guard configuration strategy.
  Every configuration in each of these strategies must induce a dominating set, therefore, they are all lower bound by the domination number $\DN$.
  We see that the following inequality holds for every graph $G$.
  \[
    \DN(G) \le \EGC(G) \le \EDN(G) \le \EDE(G)
  \]
\end{observation}

Note that we can prove bounds on all of these strategies by showing that for $G$ and its reduction $I$ it holds that $\EDE(G) \leq \EDE(I)+k$ and $\EGC(I) \leq \EGC(G)-k$ for some integer constant $k$.
If we have an exact result for $I$ the reduction gives us an exact bound on $G$ as well.
This is summed up in the following lemma.

\begin{lemma}\label{lem:technique}
  Let us assume that for graphs $G$, $I$, and an integer constant $k$
  \begin{align}
    \label{ub} \EDE(G) &\leq \EDE(I) + k,\\
    \label{lb} \EGC(G) &\geq \EGC(I) + k,\\
    \label{ih} \EDE(I) &= \EGC(I).
  \end{align}
  Then $\EDE(G) = \EGC(G)$.
\end{lemma}
\begin{proof}
  Given the assumptions, we get $\EDE(G) \leq \EGC(G)$ in the following manner.
  \[
    \EDE(G) \leq^{(\ref{ub})} \EDE(I) + k =^{(\ref{ih})} \EGC(I) + k \leq^{(\ref{lb})} \EGC(G)
  \]
  Recall Observation~\ref{obs:edn-lb} where we saw that $\EGC(G) \leq \EDE(G)$ holds, giving us the desired equality.
\qed\end{proof}

Let us have a graph $G$ with a strategy.
By \emph{simulating} a vertex attack, a vertex eviction, or an edge eviction on $G$, we mean performing the attack on $G$ and retrieving the strategy's response configuration.
Simulating attacks is useful mainly in merging several strategies over subgraphs into a strategy for the whole graph.

In the following theorem we introduce a general upper bound applicable to the m-eternal dominating strategy with eviction.

\begin{lemma}\label{lem:articulation}
  Let $G$ be a graph with an articulation $v$ such that $G \setminus v$ has two connected components $H$ and $I'$ such that there are exactly two vertices $u$ and $w$ in $V(I')$ which are neighbors of $v$.
  Let $I= \big(V(I'), E(I') \cup \{\{u,w\}\}\big)$.
  If $\{u,w\}$ lies on a cycle in $I$ then
  \[
    \EDE(G) \le \EDE(H) + \EDE(I).
  \]
\end{lemma}

\begin{proof}
  We will show that having two separate strategies for $H$ and $I$ we can merge them into one strategy for $G$ without using any additional guards.
  \begin{figure}[h]
    \centering
    \includegraphics{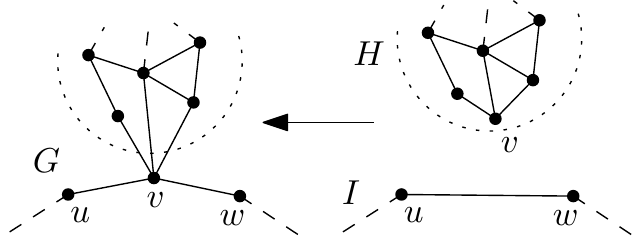}
    \caption{Decomposition of $G$ into $H$ and $I$}%
    \label{fig:upperbound}
  \end{figure}

  We will create a strategy which keeps the invariant that in all its configurations either $v$ is occupied by a guard or the pair of vertices $u$ and $w$ are evicted.
  This will ensure that whenever $v$ is not occupied due to the strategy of $H$ needing a guard from $v$ to defend other vertices of $H$, the strategy of $I$ will be in a configuration where no guard can traverse the $\{u,w\}$ edge.

  Let the initial configuration be a combination of a configuration of $H$ which defends $v$ and a configuration of $I$ which evicts $\{u,w\}$.
  The final strategy will consist of configurations which are unions of configurations of $H$ and $I$ which we choose in the following manner.

  The vertices of $G$ were partitioned among $H$ and $I$ so a vertex attack can be distinguished by the target component.
  Whenever a vertex $z$ of $G$ is attacked, choose a configuration of respective component which defends $z$.
  If $H$ was not attacked then simulate an attack on $v$.
  If $I$ was not attacked then simulate an edge eviction on $\{u,w\}$.
  By the configuration of the non-attacked component we ensure the invariant is true.
  Whenever the $\{u,w\}$ edge might be traversed by a guard in the $I$'s strategy,
  we use the fact that $v$ is occupied and instead of performing $u \rightarrow w$ we move the guards $u \rightarrow v$ and $v \rightarrow w$ which has the same effect considering guard configuration of $I$.

  The eviction of vertices and edges present in $H$ and $I$ is solved in the same way as vertex attacks.
  The only attack which remains to be solved is an edge eviction of $\{u,v\}$ or $\{w,v\}$.
  Both of these are defended by simulating an eviction of $v$ in $H$ and $\{u,w\}$ in $I$.
  The two strategies will ensure there are no guards on either $u$, $v$, nor $w$ and the invariant is still true.
\qed\end{proof}

\iffullpaper
Let $G/e$ for $e\in E(G)$ denote contraction of the edge $e$ in $G$.
We show that contracting an edge will not break an m-eternal guard configuration strategy.
\begin{lemma}\label{lem:contraction}
  Let $G$ be a graph and $e$ be its edge.
  Then for a graph $I = G / e$ ($G$ after contraction of $e$)
  \[
    \EGC(I) \le \EGC(G).
  \]
  Moreover, there is a strategy which differs only in vertices incident to $e$.
\end{lemma}
\begin{proof}
  Let $\{u,v\} = e$ and let $w\in V(I)$ be the vertex after the contraction of $e$.
  Let the m-eternal guard configuration for $I$ be the m-eternal guard configuration of $G$ where in every configuration each $u$ and $v$ is substituted with $w$.
  Any movement of guards in the original strategy is still possible in the new one.
  The only change is that instead of moving among $u$ and $v$ the guard will stay on $w$.
  Hence the traversable configurations will stay traversable.
  We devised a strategy for $I$ using the same number of guards as was used in the original strategy for $G$.
\qed\end{proof}
\fi

We may now proceed with the reductions.

\begin{lemma}\label{lem:cycle}
  Let $C_k$ be a cycle on $k$ vertices. Then $\DN(C) = \EGC(C) = \EDN(C) = \EDE(C) = \ceil[{\big}]{\frac{k}{3}}$.
\end{lemma}

\iffullpaper
\bigskip
  \begin{proof}
  It is easy to see, that the domination number of a cycle on $k$ vertices is $\big\lceil\frac{k}{3}\big\rceil$ hence it suffices to show that the m-eternal domination with eviction number is $\big\lceil\frac{k}{3}\big\rceil$ and by Observation~\ref{obs:edn-lb} we get the desired equality.

  First, note that, in every dominating configuration, two guards can be at most three edges apart.
  Any vertex attack can be defended by moving all guards in the configuration along the cycle in one direction or the other.
  \begin{figure}[h]
    \centering
    \includegraphics{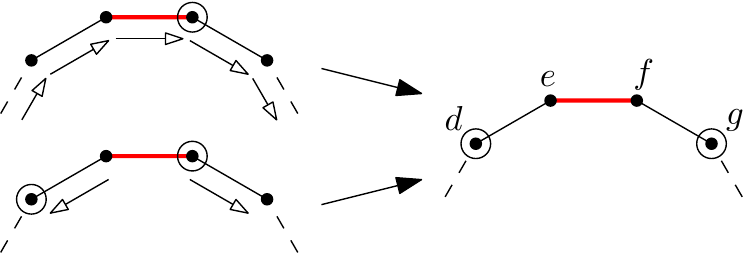}
    \caption{Two cases which can happen while evicting an edge on a cycle.}%
    \label{fig:cycle}
  \end{figure}

  Evicting a vertex can be done by evicting an incident edge since every vertex is on the cycle.
  To evict an edge $\{e,f\}$ we look at guards on the four closest vertices $(d,e,f,g)$ (in order), see Figure~\ref{fig:cycle}.
  If one incident pair of these vertices is not guarded we can rotate the configuration to move the guard gap over to $\{e,f\}$.
  If that is not the case perform $e \rightarrow d$ if $e$ is guarded and $f \rightarrow g$ if $f$ is guarded.
  Note that the new configuration has guards on $d$ and $g$ because if there was no guard on $d$ and $e$ then we would rotate the configuration (and similarly for $f$ and $g$).
  We did not move any guard which dominates vertices from the rest of the graph, and the four closest vertices are dominated by $d$ and $g$, hence the final configuration is dominating.
\qed\end{proof}
\fi

\begin{restatable}{reduction}{reductionleafedge}%
  \label{rdc:leafedge}
  Let $G$ be a Christmas cactus graph and $u$ be a leaf vertex which is connected to a vertex $v$ of degree $2$.
  Remove $u$ and $v$ from $G$.
\end{restatable}
\begin{lemma}\label{lem:leafedge}
  Let $G$ be a graph satisfying the prerequisites of Reduction~\ref{rdc:leafedge}.
  Let $I$ be $G$ after application of Reduction~\ref{rdc:leafedge}.
  Then $I$ is a Christmas cactus graph and $\EDE(G) = \EGC(G) = \EGC(I) + 1 = \EDE(I) + 1$.
\end{lemma}

\iffullpaper
\bigskip
\begin{proof}
  Let us take a graph $G$ with a vertex $x$ of degree $2$ which is connected to a leaf vertex $y$ and an articulation $v$ of $G$.
  Let $I$ be $G$ where vertices $x$ and $y$ are removed.

  Assume we have an optimal m-eternal domination strategy for $I$.
  Extend it by adding one guard on $\{x,y\}$ who can both defend and evict its vertices.
  This guard does not interfere with the rest of the strategy in any way.
  Therefore $\EDE(G)\le\EDE(I)+1$.

  Assume we have an optimal m-eternal guard configuration for $G$.
  For $y$ to be dominated there needs to be a guard on either $x$ or $y$.
  Remove vertices $x$ and $y$ along with the guard which is always present there from $G$ to obtain $I$ and move any excess guards to $v$.
  All remaining guards can substitute any movement among $v, x,$ and $y$ by staying on $v$.
  Therefore $\EGC(I) \le \EGC(G)-1$.

  Having the two inequalities we use the Lemma~\ref{lem:technique} to obtain the desired equality.
\qed\end{proof}
\fi

\begin{restatable}{reduction}{reductionleafclique}%
  \label{rdc:leafclique}
  Let $G$ be a Christmas cactus graph and $u$ be a leaf vertex which is connected to a vertex $v$ of degree $3$.
  Let the vertex $v$ has neighbors $u, x, y$ such that $x$ and $y$ are not connected.
  Remove $u$ and $v$ from $G$ and connect $x,y$ by an edge.
\end{restatable}

\begin{lemma}\label{lem:leafclique}
  Let $G$ be a graph satisfying the prerequisites of Reduction~\ref{rdc:leafclique}.
  Let $I$ be $G$ after application of Reduction~\ref{rdc:leafclique} on $u,v,x$ and $y$.
  Then $I$ is a Christmas cactus graph and $\EDE(G) = \EGC(G) = \EGC(I) + 1 = \EDE(I) + 1$.
\end{lemma}

\iffullpaper
\bigskip
\begin{proof}
  Let us take a graph $G$ with a leaf vertex $u$ which is connected to $G$ by an articulation $v$.
  Denote the vertices incident to $v$ on the Christmas cactus graph's cycle $x$ and $y$.

  \upperbound%
  We show an upper bound using the m-eternal domination with edge eviction.
  We split $G$ into two components: the leaf vertex $u$ and a graph $I= ( V(G) \setminus \{u\}, E(G) \cup \{\{u,w\}\} )$.
  We also see that $\EDE$ of a single vertex is 1.
  \begin{figure}[h]
    \centering
    \includegraphics{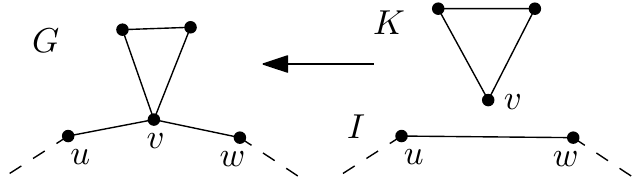}
    \caption{Decomposition of $G$ into $K$ and $I$.}%
    \label{fig:clique-upperbound}
  \end{figure}
  Use Lemma~\ref{lem:articulation} on decomposition of $G$ into $K$ and $I$ to get $\EDE(G) \le \EDE(I) + 1$.

  \lowerbound%
  Now we show that $\EGC(G) \ge \EGC(I)+1$.

  Assume that we have an optimal strategy for m-eternal guard configuration on $G$.
  Note that we can obtain $I$ by contracting $\{u, v\}$ and $\{x,v\}$ while adapting the strategy as shown in Lemma~\ref{lem:contraction}.
  We obtain a strategy for $I$ which uses exactly $\EGC(G)$ guards.
  Additionally, we see that if $u$ is to be dominated in $G$ there needs to be at least one guard in $\{u, v\}$.
  This means at least one guard is present at $x$ in all configurations of strategy on $I$.
  \begin{figure}[h]
    \centering
    \includegraphics{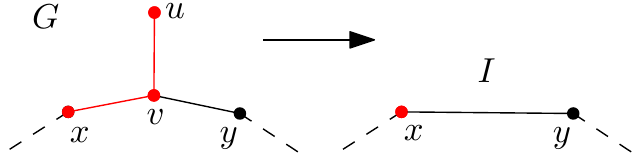}
    \caption{Contraction of $G$ edges and removal if a $K$ guard}%
    \label{fig:clique-lowerbound}
  \end{figure}

  The only case when the guard on $\{u, v\}$ needed to move from $\{u, v\}$ to $G\setminus \{u, v\}$ is when he was immediately substituted, i.e., $x \rightarrow v$ and $v \rightarrow y$ moves were performed.
  The contraction transforms this move to $x\rightarrow x$ and $x \rightarrow y$, however, in $I$ the guard from $x$ can move directly to $y$ not using the guard on $\{u, v\}$.
  Since we have a stationary guard at $x$ which is not critical for defending any vertices of $I$ we can remove him from all the configurations of the strategy for $I$. 
  Therefore $\EGC(I) \le \EGC(G)-1$.
\qed\end{proof}
\fi

\begin{restatable}{reduction}{reductionctk}%
  \label{rdc:c3k}
  Let $G$ be a Christmas cactus graph and $C$ be a leaf cycle on $n$ vertices where $n\in\{3k,3k+2\mid k\ge 1\}$.
  Let $v$ be the only articulation on this cycle.
  Remove $C \setminus v$ and create a new vertex $u$ and the edge $\{v,u\}$ in $G$.
\end{restatable}

\begin{lemma}\label{lem:c3k}
  Let $G$ be a graph satisfying the prerequisites of Reduction~\ref{rdc:c3k}.
  Let $I$ be $G$ after application of Reduction~\ref{rdc:c3k} with $C$.
  Then $I$ is a Christmas cactus graph and $\EDE(G) = \EGC(G) = \EGC(I) + k-1 = \EDE(I) + k-1$.
\end{lemma}

\iffullpaper
\bigskip
\begin{proof}
  First, we will show the bounds for the strategies and use Lemma~\ref{lem:technique} to get the desired result for $C_{3k}$.
  After the proof we will briefly discuss that exactly the same proof technique is applicable to $C_{3k-1}$.

  Let $K$ be $K_2$ present in $I$ which is to be substituted with $C_{3k}$ for $k\ge 1$.
  Let the vertices of $K$ be denoted $v$ for the articulation and $x$ for the leaf.

  \upperbound%
  We will show that the m-eternal dominating set with edge eviction strategy for $I$ can be extended to $G$ adding $k-1$ guards.

  Let $x',v',y'$ be vertices in $C_{3k}$.
  Let us create a strategy for defending $G$ by merging the strategy defending $I$ with the strategy defending $C_{3k}$ from Lemma~\ref{lem:cycle}.
  We will keep an invariant that at least one of the following pairs is occupied: $\{x,x'\}$, $\{x,y'\}$, or $\{v,v'\}$, see Figure~\ref{fig:c3k-upperbound}.
  \begin{figure}[h]\centering
    \includegraphics{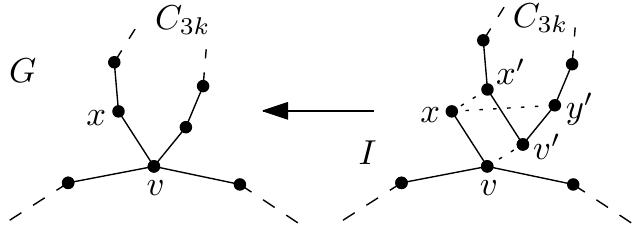}
    \caption[Reduction~\ref{rdc:c3k}]{%
      Merging strategies for $I$ and $C_{3k}$ into a united strategy for $G$.
    }%
    \label{fig:c3k-upperbound}
  \end{figure}

  Let us divide the attacks on $G$ into ones which target the vertices and edges of $C_{3k}$ and the rest which target vertices and edges present in $I$.

  If the cycle configuration changed while defending an attack, we will simulate the vertex attack on $v$ if $v'$ is occupied, or the eviction of $v$ if either $x'$ or $y'$ is occupied.
  On the other hand if configuration of $I$ changed while defending an attack on $I$ we simulate the vertex attack on $x'$ if $x$ is occupied or on $v'$ if $x$ is not occupied.
  Note that the invariant is always met so when the strategies are merged (note that $x$ is represented by either $x'$ or $y'$) we can remove the guard from strategy of $I$ which is paired up with a guard from $C_{3k}$.
  Since the $C_{3k}$ strategy of Lemma~\ref{lem:cycle} does not use two guards on incident vertices we cannot have a pair with $x$ and $v$ present at the same time. Therefore, the merged strategies do not have multiple guards at the same vertex.

  The total number of guards is $\EDE(G) \le \EDE(I)+\EDE(C_{3k})-1 \le \EDE(I)+k-1$.

  \lowerbound%
  We will show that the m-eternal guard configuration strategy on $G$ can be reduced to $I$ removing $k-1$ guards.

  Assume we have an optimal strategy for the m-eternal guard configuration on $G$.
  The leaf $C_{3k}$ must be dominated by at least $k$ guards on distinct positions.
  Let us label vertices next to $v$ on the $C_{3k}$ as $(x',x,v,v')$ (in order), see Figure~\ref{fig:c3k-lowerbound}.
  Let $P = C_{3k} \setminus \{v\}$ be a path and let $Q = C_{3k} \setminus \{x',x,v,v'\}$ be a path on $3k-4$ vertices.
  $Q$ has to be dominated by at least $k-1$ guards on $P\setminus\{x\}$ in each configuration of the strategy.
  Let us alter all the configurations by contracting all the edges of $P$ obtaining a vertex $x$.
  All the configurations in the strategy of $I$ contain $k-1$ guards on $x$
  because $P\setminus\{x\}$ was occupied by $k-1$ guards in order to dominate $Q$.
  These $k-1$ guards are always present in $P$ and translate to $k-1$ guards $x$ in $I$.
  Even if one of these guards leaves, another one must enter.
  This translates to swapping guards between $v$ and $x$ in $I$ which is an excess movement and we substitute it with $2$ stationary guards.
  Since these $k-1$ stationary guards are not necessary for dominating $I$ we remove them from $x$ resulting in the final strategy for $I$.
  \begin{figure}[h]\centering
    \includegraphics{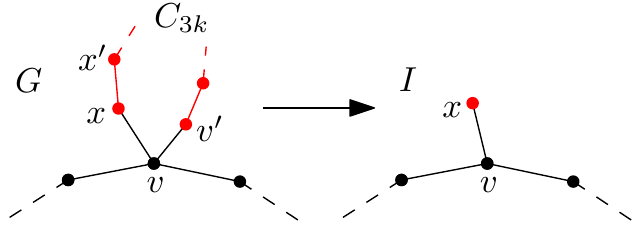}
    \caption[Reduction~\ref{rdc:c3k}]{%
      Situation in the graph $G$ in the case of Reduction~\ref{rdc:c3k}.
      On the left is the graph $G$ before edge contraction.
      On the right is the resulting graph $I$ after the edges have been contracted.
    }%
    \label{fig:c3k-lowerbound}
  \end{figure}

  Having a proof of the exact bound for $C_{3k}$ in hand, let us discuss that the same argument works for $C_{3k-1}$.
  The upper bound can change in a way that it is possible for the $C_{3k-1}$ strategy to occupy $x'$ and $y'$ at the same time, and still cannot occupy $v'$ at the same time as $x'$ or $y'$.
  This does not break the invariant because if $x'$ and $y'$ are occupied both the attacks are simulated on $x$ in $I$.
  The lower bound changes in a way that $Q$ is a path on $3k-5$ vertices which still has to be dominated by at least $k-1$ on $P\setminus\{x\}$.
\qed\end{proof}
\fi

\begin{restatable}{reduction}{reductionctkplusone}%
  \label{rdc:c3kplus1}
  Let $G$ be a Christmas cactus graph and $C$ be a leaf cycle on $3k+1$ vertices.
  Let $v$ be the only articulation on this cycle.
  Remove $C \setminus v$ from $G$.
\end{restatable}

\begin{lemma}\label{lem:c3kp1}
  Let $G$ be a graph satisfying the prerequisites of Reduction~\ref{rdc:c3kplus1}.
  Let $I$ be $G$ after application of Reduction~\ref{rdc:c3kplus1} substituting $C$ by $K_1$.
  Then $I$ is a Christmas cactus graph and $\EDE(G) = \EGC(G) = \EGC(I) + k = \EDE(I) + k$.
\end{lemma}

\iffullpaper
\bigskip
\begin{proof}
  We will show the bounds on the strategies and use the Lemma~\ref{lem:technique} to get the desired result.

  Let $G$ be the Christmas cactus graph and $C$ be its leaf cycle on $3k+1$ vertices.
  Let $v$ be the vertex of $C$ of degree bigger than $2$.
  Let $I$ be $G$ where $C \setminus \{v\}$ is removed.

  \upperbound%
  We extend an optimal m-eternal domination with edge eviction strategy on $I$ to $G$ using $k$ guards.

  Note that $v$ is an articulation and that $G \setminus \{v\}$ has two connected components.
  Let $u$ and $w$ be neighbors of $v$ on the cycle $C$.
  Let $P = C \setminus \{v\}$.
  Let us define $H = G \setminus P$ and $I' = P$.
  Observe that $I'$ as defined in Lemma~\ref{lem:articulation} is a cycle on $3k$ vertices.
  We apply Lemma~\ref{lem:articulation} on decomposition $H$ and $I$ of $G$ to get
  \[
    \EDE(G) \le \EDE(I) + \EDE(H) \le \EDE(I) + \EDE(C_{3k}) \le \EDE(I) + k.
  \]

  \lowerbound%
  We will show that $\EGC(I) \leq \EGC(G) - k$.
  Let $P = C \setminus \{v\}$ be a path on $3k$ vertices.
  Consider an optimal m-eternal guard strategy $S_G$.
  Let us alter this strategy by contracting all edges of $C$ resulting in a vertex $v'$ of $I$,
  altering the strategy due to Lemma~\ref{lem:contraction}.
  In all configurations of the new strategy the vertex $v'$ is occupied by at least $k$ guards because $P$ must have contained at least $k$ guards to be dominated.
  If we assume that the guards move only if necessary to allow change between configurations then we have $k$ stationary guards on $v'$ who never move away.
  However they do not defend any attacks because they only defended $P$ so we can remove them from every configuration of the strategy for $I$ obtaining a strategy which uses at most $\EGC(G)-k$ guards.
\qed\end{proof}
\fi

\begin{restatable}{reduction}{reductionbull}%
  \label{rdc:bull}
  Let $G$ be a Christmas cactus graph and $C$ be a cycle on three vertices $\{v,x,y\}$,
  let $x',y'$ be leafs in $G$, such that $x'$ connects to $x$, $y'$ to $y$,
  and $v$ is connected to the rest of the graph (no other edges are incident to $C$).
  We call the $C\cup\{x',y'\}$ subgraph a \emph{bull graph}.
  Remove $\{x, y, x', y'\}$ from $G$.
\end{restatable}

\begin{lemma}\label{lem:bull}
  Let $G$ be a Christmas cactus graph.
  Let $K$ be a bull graph connected to the rest of $G$ via a vertex of degree $2$.
  Let $I$ be $G$ after application of Reduction~\ref{rdc:bull} on $K$.
  Then $I$ is a Christmas cactus graph and $\EDE(G) = \EGC(G) = \EGC(I) + 2 = \EDE(I) + 2$.
\end{lemma}

\iffullpaper
\bigskip
\begin{proof}
  Let $G$ be a Christmas cactus graph and $K$ be a leaf bull graph connected to the rest of $G$ by its vertex $v$.
  Let $K' = K \setminus \{v\}$ and $I = G \setminus K'$.

  \upperbound%
  Assuming an optimal m-eternal dominating strategy on $I$ we can extend it to $G$ by adding two guards who only guard one edge each to the bull graph.
  One guard will be on $x$ or $x'$ and one on either $y$ or $y'$.
  Attacks on $I$ will be resolved by the strategy for $I$ and attacks on vertices of $K'$ will be resolved by the new guards.

  Suppose that $v$ has a neighbor $u$ not in $K'$.
  Then the edge evictions of the new cycle $\{x,v,y\}$ are resolved by simulating a vertex eviction on $\{v\}$ in $I$ and moving the new guards into $x'$ and $y'$.

  If there is no neighbor of $v$ outside of $K'$, then $G$ is exactly the bull graph and is solved as a trivial case.

  Therefore $\EDE(G)\le\EDE(I)+2$.
  \begin{figure}[h]\centering
    \includegraphics{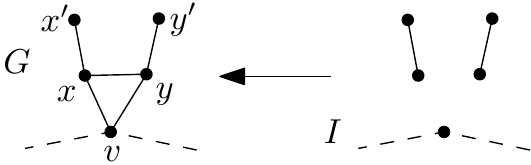}
    \caption[Reduction~\ref{rdc:bull}]{%
      Adding a guard on $\{x,x'\}$ and $\{y,y'\}$ suffices to defend the bull subgraph.
    }%
    \label{fig:bull-upperbound}
  \end{figure}

  \lowerbound%
  Let us contract all edges of the $K$ subgraph of $G$ to obtain $I$ and alter the configurations in accordance to Lemma~\ref{lem:contraction}.
  We note that the vertex $v'$ inherited all the guards of the original $K'$ which must contain at least $2$ guards in all configurations.
  Two guards from $K'$ will be stationary at $v'$ and are not necessary for defending $I$.
  This shows that $\EGC(I)\le \EGC(G)-2$.
  \begin{figure}[h]\centering
    \includegraphics{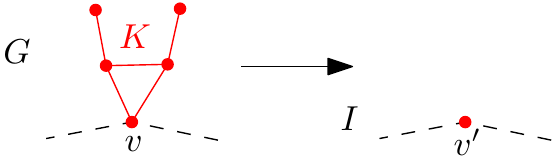}
    \caption[Reduction~\ref{rdc:bull}]{%
      The bull subgraph is contracted to $v'$ and contains $2$ guards who are no longer necessary.
    }%
    \label{fig:bull-lowerbound}
  \end{figure}
\qed\end{proof}
\fi

Let the \emph{$3$-pan graph} be a $K_3$ with one leaf attached.
\begin{restatable}{reduction}{reductionpan}%
  \label{rdc:pan}
  Let $G$ be a Christmas cactus graph and $C$ be a cycle on three vertices $\{v,x,y\}$,
  let $x'$ be a leaf in $G$, such that $x'$ connects to $x$ and $v$ is connected
  to the rest of the graph (no other edges are incident to $C$).
  The $C\cup\{x'\}$ is a $3$-pan graph.
  Remove $\{x,x'\}$ from $G$.
\end{restatable}

\begin{lemma}\label{lem:pan}
  Let $G$ be a Christmas cactus graph.
  Let $K$ be a $3$-pan graph connected to rest of the graph via a vertex of degree $2$.
  Let $I$ be $G$ after application of Reduction~\ref{rdc:pan} on $K$.
  Then $I$ is a Christmas cactus graph and $\EDE(G) = \EGC(G) = \EGC(I) + 1 = \EDE(I) + 1$.
\end{lemma}

\iffullpaper
\bigskip
\begin{proof}
  Let $G$ be a Christmas cactus graph and $K$ be a leaf $3$-pan connected to the rest of the graph by its $2$ degree vertex $v$.
  Let $K' = K \setminus \{v\}$ and $I = G \setminus K'$.

  \upperbound%
  Let us have an optimal m-eternal dominating strategy with edge eviction for $I$.
  We will extend this strategy to $G$ by adding a guard to defend the edge $\{x,x'\}$.
  These strategies cover all vertex attacks, however, since by adding vertices $x$ and $x'$ we created a cycle $\{y,v,x\}$, we have to show how to perform edge evictions on the cycle edges.
  If any edge eviction occurs, move the new guard to $x'$.
  If the edge eviction targets $\{x,y\}$ or $\{x,v\}$ then simulate eviction on the vertex $y$ or $x$ in $I$, respectively.
  If $\{v,y\}$ should be evicted, then simulate eviction on the vertex $v$ in $I$,
  which would force a guard to stand at $y$ in the next configuration. We move him to $x$ instead.
  This can be done because $x$ is a neighbor to $y$ and its only neighbor where the guard could have come from.
  Even though the guard is at $x$ in the next move we can think of him as if he was at $y$ because he has the same set of possible moves.
  \begin{figure}[h]\centering
    \includegraphics{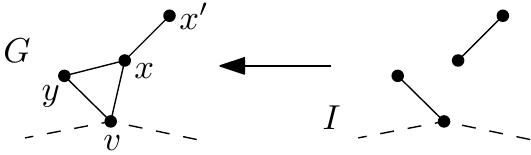}
    \caption[Reduction~\ref{rdc:pan}]{%
      Adding a guard at ${x,x'}$ and changing behavior of the guard who originally occupied either $y$ or $v$ yields a strategy for $G$.
    }%
    \label{fig:pan-upperbound}
  \end{figure}

  \lowerbound%
  Let us contract the edges of $K'$ in $G$ resulting in a single vertex $y'$ in $I$, see Figure~\ref{fig:pan-lowerbound}.
  Since $x'$ can be dominated only by a guard at $x$ or $x'$,
  then $y'$ contains a stationary guard which is not crucial for its defense.
  We remove this guard from all configurations of the strategy to obtain strategy which uses at most one less guard than the optimal guard configuration of $G$. Therefore $\EGC(I)\le\EGC(G)-1$.
  \begin{figure}[h]\centering
    \includegraphics{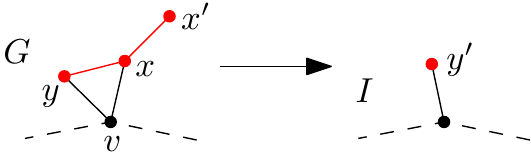}
    \caption[Reduction~\ref{rdc:pan}]{%
      Since one guard is always present at either $x$ or $x'$ in $G$, he can be removed to get strategy for $I$.
    }%
    \label{fig:pan-lowerbound}
  \end{figure}
\qed\end{proof}
\fi

Using the reductions we are ready to prove the part of Theorem~\ref{thm:cactus-algorithm} stating
that $\EGC(G) = \EDN(G) = \EDE(G)$ for all Christmas cactus graphs.

A \emph{block} or a \emph{$2$-connected component} of graph $G$ is a maximal $2$-connected subgraph of $G$.

\begin{lemma}\label{lem:leafblock}
  In a non-elementary christmas cactus graph with no leaf cycles, no leaf vertices connected to a vertex of degree $2$, and no leaf vertices connected to a block of size bigger than $3$, there is at least one leaf bull or one leaf $3$-pan graph.
\end{lemma}
\begin{proof}
  Let us call the blocks of size $3$ triangles.
  Removing edges of all the triangle subgraphs would split the christmas cactus into connected components of blocks.
  Let us choose a triangle and traverse the graph in the following way.
  If the current triangle is a bull or a $3$-pan we end the traversal and have a positive result.
  Otherwise, choose the component we have not visited yet and find a different triangle graph incident to it.
  Such triangle must exist otherwise it would be a leaf component.
  Mark this component as visited and repeat the process.
  See Figure~\ref{fig:triangle-decomposition}.
\qed\end{proof}

\begin{minipage}{0.40\textwidth}
  \includegraphics[width=1.0\textwidth]{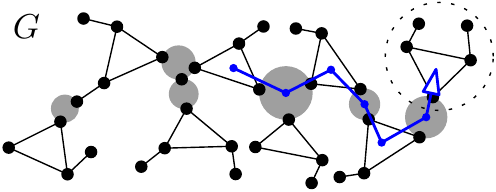}
  \captionof{figure}{%
    Process performed during the proof of an existence of a leaf bull or $3$-pan graph.
    Gray discs represent blocks.
  }%
  \label{fig:triangle-decomposition}
\end{minipage}
~
\begin{minipage}{0.50\textwidth}
  \includegraphics[width=1.0\textwidth]{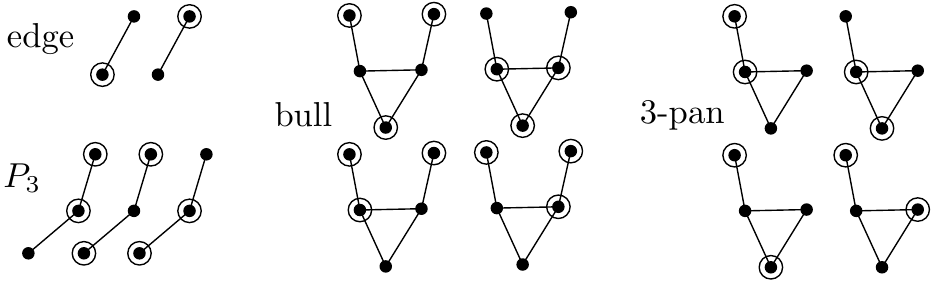}
  \captionof{figure}{%
    Mutually traversable configurations of optimal strategies for all elementary graphs except of cycles.
  }%
  \label{fig:elementary}
\end{minipage}

\begin{theorem}
  \label{thm:cactus-equality}
  Let $G$ be a Christmas cactus graph.
  Then $\EGC(G) = \EDN(G) = \EDE(G)$.
\end{theorem}

\begin{proof}
  A Christmas cactus graph $G$ always contains either a leaf or a leaf cycle.
  This will be shown by contradiction.
  If all vertices have degree at least two and each cycle has at least two neighboring blocks then the chain of blocks would either never end or it must close itself, creating another big cycle, contradicting that the graph is a cactus.

  We will use reductions until we obtain an elementary graph for which the optimal strategy is known.
  The graph is called \emph{elementary} if it is a cycle, single edge, a path on three vertices, a bull, or a $3$-pan.
  The proper reduction is chosen repeatedly in the following manner, which is also depicted in Figure~\ref{fig:allcases}.
  \begin{itemize}
    \item If $G$ is elementary we return the optimal strategy.
    \item If there is a leaf cycle on $k\ge 3$ vertices:
      \begin{itemize}
        \item If $k\not\equiv 1 \mod 3$ use Reduction~\ref{rdc:c3k},
        \item otherwise $k\equiv 1 \mod 3$ and then use Reduction~\ref{rdc:c3kplus1}.
      \end{itemize}
    \item Otherwise there is a leaf vertex $u$ in $G$ and its neighbor $v$ is an articulation.
    \item If the vertex $v$ is connected to the rest of the graph by only one edge then use Reduction~\ref{rdc:leafedge},
    \item Vertex $v$ is connected to two vertices $x$ and $y$ which are different from $u$.
    \item If there is no edge between $x$ and $y$ then they must be connected by a path in $G$, otherwise, $v$ would be in more than two blocks.
      Use Reduction~\ref{rdc:leafclique}.
    \item If there is an edge between $x$ and $y$ then it cannot be on any other cycle than $\{v,x,y\}$.
      Note that vertices $v,x,y$ form a triangle which is be connected to at most $3$ other blocks.
    \item Now, the previous evaluation can be done on every leaf vertex $u$.
      If no of the previous cases is applicable, it means by Lemma~\ref{lem:leafblock} that there is a leaf bull (use Reduction~\ref{rdc:bull}) or a leaf $3$-pan graph (use Reduction~\ref{rdc:pan}).
  \end{itemize}
  \begin{figure}[h]
    \def\widpic{12em}
    \def\picx{26em}
    \def\dstp{*-6em}
    \def\wstp{*12em}
    \centering
    \scalebox{0.9}{
    \begin{tikzpicture}[node distance = 2cm, auto]
      \node [block] (cycle)        at (0\wstp,0\dstp) {$G$ has a leaf $C_n$?};
      \node [block] (parity)       at (1\wstp,0\dstp) {$n \bmod 3$?};
      \node [pic]   (c3kp1)          at (\picx ,0\dstp) {\includegraphics[width=\widpic-1em]{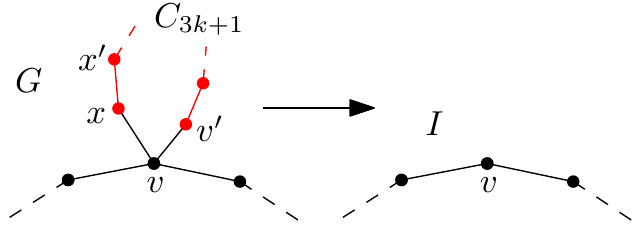}\\Reduction~\ref{rdc:c3kplus1}};
      \node [pic]   (c3k)        at (\picx ,1\dstp) {\includegraphics[width=\widpic-1em]{images/c3k-lowerbound.pdf}\\Reduction~\ref{rdc:c3k}};
      \node [block] (leaf)         at (0\wstp,1\dstp) {There is a leaf $u$ connected to $v$.};
      \node [block] (degree)       at (0\wstp,2\dstp) {Vertex $v$ has degree $2$?};
      \node [pic]   (red1)         at (\picx ,2\dstp) {\includegraphics[width=\widpic]{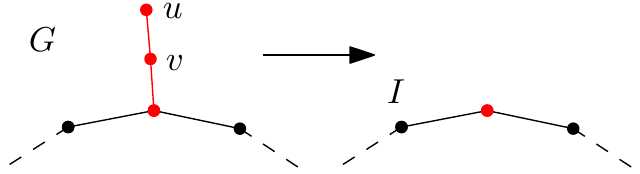}\\Reduction~\ref{rdc:leafedge}};
      \node [block] (isconnected)  at (0\wstp,3\dstp) {Are neighbors $x$ and $y$ of $v$ connected?};
      \node [pic]   (notconnected) at (\picx ,3\dstp) {\includegraphics[width=\widpic]{images/clique-lowerbound.pdf}\\Reduction~\ref{rdc:leafclique}};
      \node [block] (bull)         at (0\wstp,4\dstp) {Is there a leaf bull in the graph?};
      \node [pic]   (redbull)      at (\picx ,4\dstp) {\includegraphics[width=\widpic]{images/bull-lowerbound.pdf}\\Reduction~\ref{rdc:bull}};
      \node [block] (pan)          at (0\wstp,5\dstp) {There is a leaf $3$-pan graph.};
      \node [pic]   (redpan)       at (\picx ,5\dstp) {\includegraphics[width=\widpic]{images/pan-lowerbound.pdf}\\Reduction~\ref{rdc:pan}};

      \path [line] (cycle) --node{no} (leaf);
      \path [line] (leaf) -- (degree);
      \path [line] (degree) --node{no} (isconnected);
      \path [line] (isconnected) --node{yes} (bull);
      \path [line] (bull) --node{no} (pan);

      \path [line] (cycle) --node{yes} (parity);
      \path [line] (parity) --node{1} (c3kp1);
      \path [line] (parity) edge[bend right] node[below left]{0 or 2} (c3k);
      \path [line] (degree) -- node{yes} (red1);
      \path [line] (isconnected) --node{no} (notconnected);
      \path [line] (bull) --node{yes} (redbull);
      \path [line] (pan) --node{Always true by Lemma~\ref{lem:leafblock}.} (redpan);
    \end{tikzpicture}
  }
    \caption{Decision tree for choosing a proper reduction}%
    \label{fig:allcases}
  \end{figure}

  Using the reductions we eventually end up in a situation where the Christmas cactus graph is an elementary graph.
  The optimal strategy for cycle was shown in Lemma~\ref{lem:cycle}, all the configurations of optimal strategies of all the remaining graphs are depicted in Figure~\ref{fig:elementary}.

  In each of these elementary graphs, allowing eviction attacks does not increase the necessary number of guards.
  Also, allowing more guards at one vertex does not add any advantage and does not decrease the necessary number of guards.
  Therefore, for all of these cases it holds that $\EDE = \EDN = \EGC$.
\qed\end{proof}

\section{Upper bound on the m-eternal domination number of cactus graphs}
\label{sec:bound}
\begin{definition}
  Let us have a cactus graph $G$.
  Let us color vertices of $G$ in the following way.
  Let a vertex be colored \emph{red} if it is contained in more than two $2$-connected components of $G$, otherwise it is colored \emph{black}.
  Let $R(G)$ denote the number of red vertices, and $Rg(G)$ denote the number of red connected components (e.g. $R(G)=7$ and $Rg(G)=3$ in Figure~\ref{fig:red-decomposition}).

  Let $G'$ be a graph created from $G$ by contracting each red connected components into a red vertex.

  Let $B_{G'}$ be a set of maximal connected components of black vertices in $G'$.
  Let $R_{G'}=\{b \cup N(b) \mid b\in B_{G'}\}$.
  Note that $N(b)$ contains only red vertices.
  Let the \emph{Christmas cactus decomposition} $Op(G)$ be a disjoint union of graphs induced by $G[r]$ for all $r\in R_{G'}$.
  See Figure~\ref{fig:red-decomposition}.
  \begin{figure}[h]\centering
    \includegraphics[width=0.95\textwidth]{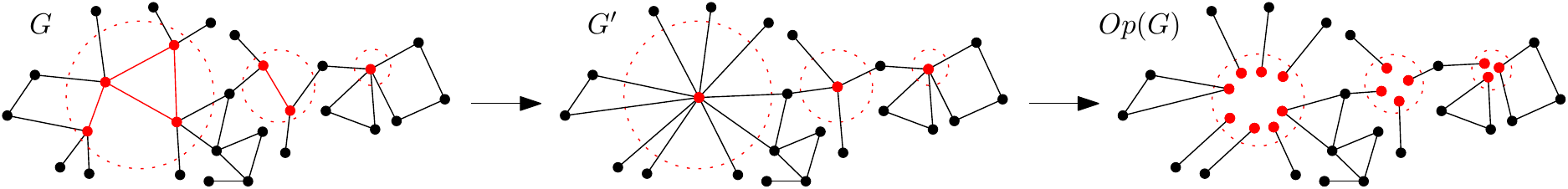}
    \caption{%
      Process of transforming a cactus $G$ by contracting red edges to produce $G'$ and subsequently duplicating red vertices for each black connected component to get $Op(G)$.
    }%
    \label{fig:red-decomposition}
  \end{figure}
\end{definition}

\begin{theorem}\label{thm:cactus-upperbound}
  The m-eternal domination number of a cactus graph $G$ is bounded by
  \[
    \EDN(G) \le \smashoperator{\sum_{H \in Op(G)}}\big(\EDN(H)-R(H)\big) + R(G) + Rg(G),
  \]
  where $Op(G)$ are the components of the Christmas cactus decomposition,
  $R(G)$ is the number of red vertices in $G$ and $Rg(G)$ is the number of connected components of red vertices.
\end{theorem}

\begin{proof}
  Let $G'$ be a graph where all connected components of red vertices are contracted creating one red group vertex for each component as shown in Figure~\ref{fig:red-decomposition}.
  Let $Op(G)$ be the Christmas cactus decomposition of $G$.

  First, find an optimal strategy for each Christmas cactus graph in $Op(G)$ separately by the process presented in Section~\ref{sec:cacti-graphs}.
  We will show how to merge these disjoint strategies into one strategy for the whole graph $G'$ and subsequently generalize it for $G$.

  Assume that all red vertices of $G'$ are always occupied, hence we have to show how to defend black vertices.
  Let us reverse the process of Christmas cactus decomposition and merge disjoint Christmas cactus graphs by the red vertices to obtain $G'$.
  When a black vertex is attacked we simulate an attack on the respective Christmas cactus graph to get a configuration which defends the vertex as shown in Fig.~\ref{fig:attack-simulation}.
  If any of the red vertices of the Christmas cactus graph are not occupied then we simulate an attack on the red vertex in all other Christmas cactus graphs which contain it.

  The process ensures that in each configuration all but one Christmas cactus graph incident to each red vertex has a guard on it.
  A red vertex $v$ incident to $k$ black components in $G'$ is always occupied by exactly $k-1$ guards.
  So we can remove $k-2$ guards and it remains always occupied by exactly one guard.
  This strategy for $G'$ uses $\sum_{H\in Op(G)}\big(\EDN(H) - R(H)\big) + 2Rg(G)$ guards.

  \begin{figure}[h]\centering
    \includegraphics[width=0.7\textwidth]{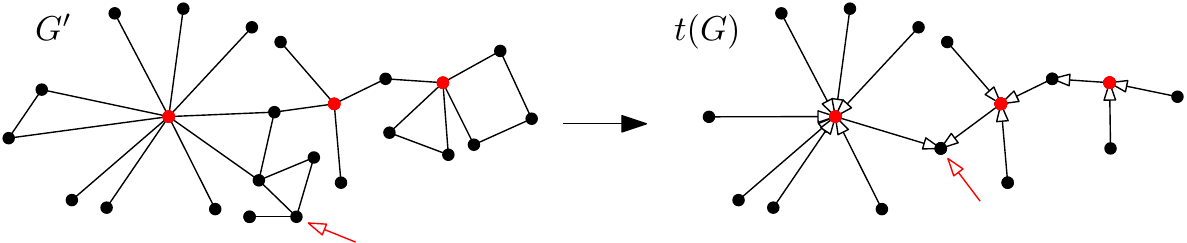}
    \caption{%
      Simulating attacks to get the right amount of guards on each red vertex.
    }%
    \label{fig:attack-simulation}
  \end{figure}

  Now we get $G$ from $G'$ by expanding the red vertices back into the original red connected components.
  Add guards such that all red vertices are occupied.
  The strategy will be altered slightly.
  When we defend $G'$ by moving a guard from red vertex $v$ then another guard from a different Christmas cactus component is forced to move to $v$ by a simulated attack.
  However in $G$ the left vertex $u$ and the attacked vertex $u'$ might not coincide so we pick a path from $u$ to $u'$ in the red component and move all the guards along the path.

  The change in number of guards can be imagined as removing guards on red vertices of $G'$ and adding guards on all red vertices of $G$.
  We devised a strategy for $G$ which uses $\sum_{H\in Op(G)}\big(\EDN(H) - R(H)\big) + R(G) + Rg(G)$ guards.
\qed\end{proof}

\section{Linear-time algorithm}
\label{sec:algorithm}

We present a description of a linear-time algorithm, which computes $\EDN$ in
Christmas cactus graphs in linear time.
The algorithm applies previously presented reductions on the block-cut tree of the input graph.

\begin{definition}[F. Harary~\cite{graph-theory}]
  Let the \emph{block-cut tree} of a graph $G$ be a graph $\BC(G) = (A \cup \mathcal{B}, E')$, where $A$ is the set of articulations in $G$ and $\mathcal{B}$ is the set of biconnected components in $G$.
  A vertex $a \in A$ is connected by an edge to a biconnected component $B \in \mathcal{B}$ if and only if $a \in B$ in $G$.
\end{definition}

The high level description of the algorithm is as follows.


\begin{enumerate}
  \item Construct the Christmas cactus decomposition of $G$ and iterate the following for each component $\spread[H,c]$.
    \begin{enumerate}
      \item Construct the block-cut tree $\BC$ of the Christmas cactus $H_i$.
      \item Repeatedly apply the reductions on the leaf components of $\BC$.
      \item If a reduction can by applied, appropriately modify $\BC$ in constant time, so that it represents $H_i$ with the chosen reduction applied.
        At the same time, increase the resulting $\EDN$ appropriately.
      \item If $\BC$ is empty return the resulting $\EDN$ and end the process.
    \end{enumerate}
  \item Use Theorem~\ref{thm:cactus-upperbound} to get the upper bound.
\end{enumerate}

This result is summed up in the following theorem.
We also present the detailed pseudocode of the linear-time algorithm for
finding the m-eternal domination number for Christmas cactus graphs.

  \newcounter{algolines} 
  \begin{algorithm}[h]
    \caption{$\EDN$ of a Christmas cactus graph, Part 1}
    \label{alg:opuntia}
    \begin{algorithmic}[1]
      \Procedure{m-EDN-Christmas-cactus}{$G$}
        \State $\BC = (V', E', size, deg) \gets \text {the block-cut tree of } G$\label{line:bc}

        \State $stack \gets \emptyset$ \Comment{$stack$ keeps track of all leaf blocks in $\BC$}
        \For{$v \in V'$}
          \If{$deg(v) \le 1$} \Comment{All leaf blocks, or the only block}
            \State add $v$ on top of the $stack$
          \EndIf
        \EndFor

        \State $g \gets 0$ \Comment{The resulting $\EDN(G)$}
        \While{$stack \neq \emptyset$}\label{line:main-loop}
          \State $v \gets$ retrieve and remove an element from top of the $stack$
          \If{$deg(v)=0$}\label{line:triv-case} \Comment{Block is an elementary cycle, an edge, or one vertex}
            \State $(g, stack, size) \gets \Call{remove-leaf-block}{v, stack}$
          \Else
            \State $a \gets$ the articulation incident to $v$ \Comment{$deg(a)=2$ in Ch. cactus graphs}
            \State $u \gets$ the second block neighbor of $a$ other than $v$
            \State $(g, stack, size) \gets \Call{block}{u, v, g, stack, size}$
          \EndIf
        \EndWhile
        \State \Return $g$
      \EndProcedure
      \setcounter{algolines}{\arabic{ALG@line}}
    \end{algorithmic}
  \end{algorithm}

  \setcounter{algorithm}{0}
  \begin{algorithm}[h!]
    \caption{$\EDN$ of a Christmas cactus graph, Part 2}
    \begin{algorithmic}[1]
      \setcounter{ALG@line}{\arabic{algolines}}
      \Procedure{block}{$u, v, g, stack, size$}\label{line:block}
        \If{$size(v) \ge 3$}\Comment{Leaf cycle of size at least $3$}
          \State $(g, stack, size) \gets \Call{leaf-cycle}{v, g, stack, size}$
        \ElsIf{$size(v) = 2$}
          \If{$size(u) = 2$}\Comment{Reduction~\ref{rdc:leafclique}}
            \State $g \gets g + 1$\label{line:leafedge}
            \State $(g, stack) \gets \Call{remove-leaf-block}{v, g, stack, size}$
            \State $(g, stack) \gets \Call{remove-leaf-block}{u, g, stack, size}$
          \ElsIf{$size(u) \ge 3$}\Comment{Partial Reductions~\ref{rdc:bull} and~\ref{rdc:pan}}
            \State $g \gets g + 1$ \label{line:contraction}
            \State $(g, stack) \gets \Call{remove-leaf-block}{v, g, stack, size}$
            \State $size(u) \gets size(u)-1$ 
          \EndIf
        \EndIf
        \State \Return $(g, stack, size)$
      \EndProcedure

      \Procedure{leaf-cycle}{$v, g, stack, size$}\label{line:leaf-cycle}
        \If{$size(v) \equiv 0 \mod 3$ or $size(v) \equiv 2 \mod 3$}\label{line:deleted-block}\Comment{Reduction~\ref{rdc:c3k}}
        \State $g \gets g + \ceil{size(v)/3} - 1$\label{line:cycle-mod1}
          \State $size(v) \gets 2$
          \State add $v$ on top of the $stack$ \Comment{The block still remains a leaf}
        \ElsIf{$size(v) \equiv 1 \mod 3$}\Comment{Reduction~\ref{rdc:c3kplus1}}
          \State $g \gets g + (size(v) - 1)/3$\label{line:cycle-mod0}
          \State $(g, stack) \gets \Call{remove-leaf-block}{v, g, stack, size}$
        \EndIf
        \State \Return $(g, stack, size)$
      \EndProcedure

      \Procedure{remove-leaf-block}{$v, g, stack, size$}
        \If{$deg(v) \ge 1$}
          \State $a \gets$ the articulation incident to $v$ \Comment{$deg(a)=2$ in Ch. cactus graphs}
          \State $u \gets$ the second block neighbor of $a$ other than $v$
          \State remove $a$ from the neighbor list in $u$ and erase $v$ and $a$ from $\BC$
          \If{$deg(u) \le 1$}
            \State add $u$ on top of the $stack$ \Comment{Vertex can become a leaf only here}
          \EndIf
        \ElsIf{$deg(v) = 0$}
          \State $g \gets g + \ceil{size(v)/3}$ \label{line:isolatedblock}
          \State erase $v$ from $\BC$
        \EndIf
        \State \Return $(g, stack)$
      \EndProcedure
    \end{algorithmic}
  \end{algorithm}

\begin{theorem}\label{thm:linear-alg}
  Let $G$ be a cactus on $n$ vertices and $m$ edges. Then there exists an algorithm which computes an upper bound on $\EDN(G)$ in time $\mathcal{O}(n+m)$.
  Moreover, this algorithm computes the $\EDN$ of Christmas cactus graphs exactly.
\end{theorem}

  \begin{proof}
    First step of the algorithm is to create the Christmas cactus decomposition of the input graph $G$.
    For each of these Christmas cactus graphs we run Algorithm~\ref{alg:opuntia} and then output the answer devised by Theorem~\ref{thm:cactus-upperbound}.

    The construction of a Christmas cactus decomposition of $G$ is done by constructing the block-cut tree of $G$ and coloring red all the vertices which are present in at least $3$ blocks, and coloring black all other vertices.
    We use the DFS algorithm to find all the connected components of red vertices and contract each of these components into a single red vertex.
    Next, we use the DFS algorithm to retrieve all connected components of black vertices along with their incident red vertices.
    Note that every edge of $G$ is contained at most once in the Christmas cactus decomposition, hence the total number of vertices and edges in the decomposition is bounded by $\mathcal{O}(|E(G)|)$.

    We run the Algorithm~\ref{alg:opuntia} for each component separately.
    Now we will show its correctness.
    The algorithm performs reductions in the while loop at line~\ref{line:main-loop}.
    In each iteration it processes one leaf block vertex in $\BC$.

    First, we argue that the algorithm correctly counts the number of guards on all elementary graphs.
    In case $\BC$ consists of a single block, it is detected on line \ref{line:triv-case}.
    The block is removed and the number of guards increases by $\ceil{size(v)/3}$ on line~\ref{line:isolatedblock}, which is consistent with the result for elementary cycle and edge.
    In the other case the $\BC$ consists of several blocks.
    If $\BC$ represents a path on three vertices, one guard will added by line~\ref{line:leafedge} and one by line~\ref{line:isolatedblock} and both blocks are removed.
    If $\BC$ represents a $3$-pan it will count $2$ guards by first reducing one block by line~\ref{line:contraction} or \ref{line:leafedge}, and then reducing a single block of size at most $2$ on line~\ref{line:isolatedblock}.
    If $\BC$ represents a bull the algorithm will find a leaf block of size $2$ reducing the bull to a path on $3$ vertices on line~\ref{line:contraction}, counting correctly $3$ guards.

    Now we show that if the algorithm performs an operation on a leaf block it resolves in the correct number of guards at the end.
    Let $v$ be a leaf block vertex processed in the loop.
    Note that each block of a Christmas cactus graph is either a cycle or an edge.

    Consider the case where $v$ is a leaf cycle.
    If the cycle has $size(v)\equiv 1\mod 3$, then the block is removed entirely adding $(size(v)-1)/3$ guards on line~\ref{line:cycle-mod1}, exactly as in Reduction~\ref{rdc:c3kplus1}.
    Otherwise the cycle is contracted to a block of size $2$ and $\ceil{size(v)/3}-1$ guards are added on line~\ref{line:cycle-mod0}, as in Reduction~\ref{rdc:c3k}.

    Consider the case where the leaf block $v$ has $size(v)=2$, representing a leaf vertex.
    Let $u$ be the block that shares an articulation with $v$.
    If block $u$ has $size(u)=2$ then we remove both of these blocks and add a guard on line~\ref{line:leafedge}, as in Reduction~\ref{rdc:leafedge}.
    If block $u$ has $size(u)\ge 3$ then $v$ is a leaf vertex connected to a cycle.
    On line~\ref{line:contraction} $v$ is removed and $size(u)$ is decreased by one.
    This is consistent with Reductions~\ref{rdc:leafclique}, \ref{rdc:bull}, and \ref{rdc:pan}.
    Note that reducing a leaf incident to a component on three vertices first, yields the same result as reducing the graph first and waiting for the leaf be reduced in either Reduction~\ref{rdc:leafclique}, \ref{rdc:bull}, or \ref{rdc:pan}.


    The algorithm performs reductions which were proved to be correct.
    This concludes the proof of correctness.

    Let $n'$ be the number of vertices of the Christmas cactus graph, and $m'$ be the number of its edges.
    Now we show that Algorithm~\ref{alg:opuntia} runs in time $\mathcal{O}(n' + m')$.
  Using Tarjan's algorithm \cite{tarjan-alg}, we can find the blocks of a graph in linear time.
  By the straight-forward augmentation of the algorithm we obtain the block-cut tree $\BC$ where every block $v$ contains additional information $size(v)$ with the number of vertices it contains.
  Note that the number of vertices and edges of $\BC$ is bounded by $2n'$.
  Therefore $|V(\BC)| = \mathcal{O}(n')$ and $|E(\BC)| = \mathcal{O}(n')$.

  Now consider the while loop at line \ref{line:main-loop}.
  We claim that every vertex in $\BC$ is processed at most twice in the loop and every iteration takes constant time.
  Let $v$ be the currently processed vertex.
  During one iteration of the main cycle either a block of size at most $2$ is deleted or a block of size at least $3$ is either shrunk to size $2$ or deleted.
  Therefore, the algorithm for Christmas cactus graphs runs in time $\mathcal{O}(n' + m')$.

  For the cactus graph we need to create the Christmas cactus decomposition which uses Tarjan's algorithm \cite{tarjan-alg} for creating the block-cut tree, and the DFS, which both runs in $\mathcal{O}(n+m)$.
  As stated, the running time of the algorithm is linear in the size of the Christmas cactus.
  Therefore, the total running time is bound by sum of their sizes $\mathcal{O}(n+m)$.
\qed\end{proof}

This result together with Theorem~\ref{thm:cactus-equality} implies Theorem~\ref{thm:cactus-algorithm}.



\section{Future work}
The computational complexity of the decision variant of the \MED\ problem is still mostly unknown as mentioned in the introduction.

The natural extension of the algorithm from cactus graphs is to the more general case of graphs with treewidth $2$.
It is also an interesting question if we can design an algorithm, whose running time is parameterized by the treewidth of the input graph.

\bigskip
\noindent
{\bf Acknowledgments.} We would like to thank Martin Balko and an anonymous referee for their valuable comments and insights.



\bibliographystyle{plain}
\bibliography{main}

\begin{thebibliography}{10}

\bibitem{proper-interval-graphs}
Andrei Braga, Cid~C. de~Souza, and Orlando Lee.
\newblock The eternal dominating set problem for proper interval graphs.
\newblock {\em Information Processing Letters}, 115(6):582--587, 2015.

\bibitem{ineq1}
Alewyn~P. Burger, Ernest~J. Cockayne, W.~R. Gründlingh, Christina~M. Mynhardt,
  Jan~H. van Vuuren, and Wynand Winterbach.
\newblock Infinite order domination in graphs.
\newblock {\em Journal of Combinatorial Mathematics and Combinatorial
  Computing}, 50:179--194, 2004.

\bibitem{3n-grids}
Stephen Finbow, Margaret-Ellen Messinger, and Martin~F. van Bommel.
\newblock Eternal domination on 3 {$\times$} n grid graphs.
\newblock {\em Australasian Journal of Combinatorics}, 61:156--174, 2015.

\bibitem{eternal-security-in-graphs}
Wayne Goddard, Sandra~M. Hedetniemi, and Stephen~T. Hedetniemi.
\newblock Eternal security in graphs.
\newblock {\em Journal of Combinatorial Mathematics and Combinatorial
  Computing}, 52:169--180, 2005.

\bibitem{intervals}
Udaiprakash~I. Gupta, Der-Tsai Lee, and Joseph Y.-T. Leung.
\newblock Efficient algorithms for interval graphs and circular arc graphs.
\newblock {\em Networks}, 12(4):459--467, 1982.

\bibitem{graph-theory}
Frank Harary.
\newblock {\em Graph Theory}.
\newblock Addison-Wesley Publishing Company, Inc., 1969.

\bibitem{bounds-meternal-dom}
Michael~A. Henning, William~F. Klostermeyer, and Gary MacGillivray.
\newblock Bounds for the m-eternal domination number of a graph.
\newblock {\em Contributions to Discrete Mathematics}, 12(2), 2017.

\bibitem{eternal-dom-sets}
William~F. Klostermeyer and Gary MacGillivray.
\newblock Eternal dominating sets in graphs.
\newblock {\em Journal of Combinatorial Mathematics and Combinatorial
  Computing}, 68:97--111, February 2009.

\bibitem{survey-article}
William~F. Klostermeyer and Christina~M. Mynhardt.
\newblock Protecting a graph with mobile guards.
\newblock {\em Applicable Analysis and Discrete Mathematics}, 10, 07 2014.

\bibitem{christmas-cactus}
Tom Leighton and Ankur Moitra.
\newblock Some results on greedy embeddings in metric spaces.
\newblock {\em Discrete {\&} Computational Geometry}, 44(3):686--705, Oct 2010.

\bibitem{tarjan-alg}
Robert Tarjan.
\newblock Depth-first search and linear graph algorithms.
\newblock {\em SIAM Journal on Computing}, 1(2):146--160, 1972.

\bibitem{5n-grids}
Christopher~M. van Bommel and Martin~F. van Bommel.
\newblock Eternal domination numbers of 5 {$\times$} n grid graphs.
\newblock {\em Journal of Combinatorial Mathematics and Combinatorial
  Computing}, 97:83--102, 2016.

\end{thebibliography}


\end{document}